\documentclass[envcountsame]{llncs}
\usepackage{amsmath,amssymb,amsbsy,amsfonts,latexsym,
               amsopn,amstext,amsxtra,euscript,amscd,commath,mathtools,braket}
\usepackage{color,soul}
\usepackage[noadjust]{cite}
\usepackage{algorithm}
\usepackage{algcompatible}

\def \F {{\mathbb F}}

\def \Z {{\mathbb Z}}

\def\F{{\mathbb F}}
\def\Z{{\mathbb Z}}

\def\F{{\mathbb F}}

\def\R{{\mathbb R}}
\def\Z{{\mathbb Z}}

\def\R{{\mathbb R}}
\def\C{{\mathbb C}}

\def\00{{\bf 0}}
\def\11{{\bf 1}}
\def\+{\oplus}

\def \F {{\mathbb F}}

\def \Z {{\mathbb Z}}

\def\wt{{\rm wt}}

\def\dist{{\rm dist}}

\begin{document}

\title{\bf  A quantum algorithm to estimate the Gowers $U_2$ norm and
linearity testing of Boolean functions}

\author{ C.~A.~Jothishwaran\inst{1}, Anton Tkachenko\inst{2}, Sugata~Gangopadhyay\inst{1}  \\
Constanza Riera\inst{2}, Pantelimon St\u anic\u a\inst{3}}
\institute{Department of Computer Science and Engineering,\\
  %\and Department of Physics \\
Indian Institute of Technology Roorkee, Roorkee 247667, INDIA\\ \email{jothi@ph.iitr.ac.in, sugata.gangopadhyay@cs.iitr.ac.in}\\
\and  Department of Computer Science, Electrical Engineering and Mathematical
Sciences, \\
Western Norway University of Applied Sciences, 5020 Bergen, NORWAY\\ \email{Anton.Tkachenko@hvl.no, csr@hvl.no}
\and Department of Applied Mathematics  \\
Naval Postgraduate School, Monterey, CA 93943--5216, USA \\ \email{pstanica@nps.edu}}

\maketitle

\abstract{We propose a quantum algorithm to estimate the Gowers $U_2$ norm of a Boolean function, and extend it into a  second algorithm  to distinguish between linear Boolean functions and Boolean functions that are $\epsilon$-far from the set of linear Boolean functions, which seems to perform better than the classical BLR algorithm. 
Finally, we outline an algorithm to estimate
Gowers $U_3$ norms of Boolean functions. }


\section{Introduction}

Gowers uniformity norms were introduced by Gowers \cite{Gowers2001} to prove Szmer\'edi's theorem. In their full generality, Gowers uniformity norms operate over functions from finite sets to the field of complex numbers. The Gowers uniformity norm of dimension $d$ of a function $f$ tells us the extent of correlation of $f$ to the polynomial phase functions of degree up to $d-1$. In this paper, we consider the Gowers uniformity norm $U_2$ of dimension $2$  
for Boolean functions, and find a quantum estimate of its upper bound. We also propose a 
linearity test of Boolean functions based on the same quantum algorithm, which seems to perform better than the classical BLR algorithm. 

\subsection{Boolean functions}

We denote the ring of integers, the set of positive integers, and the fields of real numbers and complex numbers by $\Z$, $\Z^+$, $\R$, and $\C$, respectively. For any $n \in \Z^+$, the set $[n] = \{ i \in \Z^+: 1 \leq i \leq n \}$, and $\F_2^n = \{ x = (x_1, \ldots, x_n) : x_i \in \F_2, \mbox{ for all } i \in [n]\}$ where $\F_2$ is the prime field of characteristic $2$. Addition in each of the above algebraic systems is denoted by `$+$'. An $n$-variable Boolean function $F$ is a function from $\F_2^n$ to $\F_2$. The set of all such functions is 
denoted by  $\mathfrak B_n$. Each function $F \in \mathfrak B_n$ has its character form $f : \F_2^n \rightarrow \R$ defined by $f(x) = (-1)^{F(x)}$, for all $x \in \F_2^n$. In this article, abusing notation, we refer to the character form $f$ as Boolean functions and go to the extent of writing $f \in \mathfrak B_n$, whenever $F \in \mathfrak B_n$, if there is no danger of confusion. For any $x, y \in \F_2^n$, the inner product $x \cdot y = \sum_{i \in [n]} x_i y_i$ where the sum is over $\F_2$. The (Hamming) weight of a vector $u=(u_1,\ldots,u_n) \in \F_2^n$ is $\wt (u) = \sum_{i \in [n]}u_i$, where the sum is over $\Z$. The weight of a Boolean function $F \in \mathfrak B_n$, or equivalently $f \in \mathfrak B_n$ is the cardinality
$\wt (F) = \abs{ \{x \in \F_2^n: F(x) \neq 0 \}}$, or equivalently
$\wt (f) = \abs{ \{x \in \F_2^n: f(x) \neq 1 \}}$. The Hamming distance
between $F, G \in \mathfrak B_n$, or equivalently, between $f, g \in \mathfrak B_n$ is $d_H(F, G) = \abs{\{ x \in \F_2^n : F(x) \neq G(x)\}}$, or, $d_H(f, g) = \abs{\{ x \in \F_2^n : f(x) \neq g(x)\}}$. Any  Boolean function $F \in \mathfrak B_n$ can be expressed as a polynomial, called the algebraic normal form (ANF),
\begin{equation}
\label{anf-def-eq}
F(x_1, \ldots, x_n) = \sum_{u \in \F_2^n} \lambda _u x^u
\mbox{ where } \lambda _u \in \F_2, \mbox{ and } x^u = \prod_{i \in [n]}x_i^{u_i}.
\end{equation}
The algebraic degree of a Boolean function $\deg(F) = \max \{ \wt(u) : \lambda_u \neq 0\}$. A Boolean function with algebraic degree at most $1$ is said to be an affine function. An affine function in $\mathfrak B_n$ is of the form $\varphi (x) = u \cdot x + \varepsilon$ for some $u \in \F_2^n$ and $\varepsilon \in \F_2$. An affine function with $\varepsilon = 0$ is said to be a linear function. We denote the set of all $n$-variable affine functions by $\mathfrak A_n$, and the set of all $n$-variable linear functions by $\mathfrak L_n$.

The Fourier series expansion of $f \in \mathfrak B_n$ is
\begin{equation}
\label{fourier-exp}
f(x) = \sum_{ u \in \F_2^n} \widehat f (u) (-1)^{u \cdot x}.
\end{equation}
The coefficients $\widehat f (u)$ are said to be the Fourier coefficients of $f$.
The transformation $f \mapsto \widehat f$ is the Fourier transformation of $f$.
It is known that
\begin{equation}
\label{ortho-linear}
\sum_{x \in F_2^n} (-1)^{v \cdot x} = \begin{cases} 0 & \mbox{ if } v \neq 0 \\
2^n & \mbox { if } v =0.
\end{cases}
\end{equation}
Equations~\eqref{fourier-exp} and \eqref{ortho-linear} yield
\begin{equation}
\label{inv-fourier}
\begin{split}
\sum_{x \in \F_2^n} f(x)(-1)^{u\cdot x} &=
\sum_{ x \in \F_2^n} \sum_{ v \in \F_2^n} \widehat f (v) (-1)^{(u+v) \cdot x}\\
& = \sum_{ v \in \F_2^n} \widehat f (v) \sum_{ x \in \F_2^n} (-1)^{(u+v) \cdot x}
= 2^n \widehat f(u),
\end{split}
\end{equation}
that is, $\widehat f(u) = 2^{-n} \sum_{x \in \F_2^n} f(x)(-1)^{u \cdot x}$.
The sum
\begin{equation*}
\begin{split}
\sum_{x \in \F_2^n} f(x)^2 &= \sum_{x \in \F_2^n} \sum_{ u \in \F_2^n} \widehat f (u) (-1)^{u \cdot x}\sum_{ v \in \F_2^n} \widehat f (v) (-1)^{v \cdot x}\\
&= \sum_{ u \in \F_2^n}\sum_{ v \in \F_2^n} \widehat f (u) \widehat f (v)
\sum_{x \in \F_2^n} (-1)^{(u + v) \cdot x} = 2^n \sum_{u\in \F_2^n}\widehat f (u)^2.
\end{split}
\end{equation*}
The identity $\sum_{x\in \F_2^n}\widehat f (x)^2 = 2^{-n}\sum_{x \in \F_2^n} f(x)^2$, is known as the {\em Plancherel's identity}.
This is true for $f : \F_2^n \rightarrow \R$. If $f \in \mathfrak B_n$,
we have the {\em Parseval's identity} $\sum_{u \in \F_2^n}\widehat f(u)^2 = 1$.
For $f, g \in \mathfrak B_n$ the convolution product, $f \ast g$ is defined as
\begin{equation}
\label{convo}
(f \ast g) (x) = 2^{-n} \sum_{y \in \F_2^n} f(y)g(x+y)= 2^{-n} \sum_{y \in \F_2^n} f(x+y)g(y).
\end{equation}
Using \eqref{inv-fourier} on \eqref{convo}
\begin{equation}
\label{convo-1}
\begin{split}
\widehat{f \ast g}(u) & = 2^{-n} \sum_{x \in \F_2^n} (f \ast g)(x)(-1)^{u \cdot x}
= 2^{-2n} \sum_{x \in \F_2^n} \sum_{y \in \F_2^n} f(y)g(x+y)(-1)^{u \cdot x} \\
&= 2^{-2n} \sum_{x \in \F_2^n} \sum_{y \in \F_2^n} f(y)(-1)^{u \cdot y}
g(x+y)(-1)^{u \cdot (x+y)} \\
&= \left(2^{-n} \sum_{y \in \F_2^n} f(y)(-1)^{u \cdot y}\right)
\left( 2^{-n} \sum_{x \in \F_2^n} g(y)(-1)^{u \cdot x} \right)
= \widehat f (u) \widehat g(u).
\end{split}
\end{equation}
For each $x \in \F_2^n$, $(f\ast f) (x) = 2^{-n} \sum_{y \in \F_2^n}f(y)f(x+y)$
is said to be the {\em autocorrelation} of $f$ at $x$, and
$\widehat{f \ast f}(x) = \widehat f (x)^2$.

The derivative of $f \in \mathfrak B_n$ at $c \in \F_2^n$ is the function
\begin{equation}
\label{deriv-defn}
\Delta_c f(x) = f(x)f(x+c), \mbox{ for all } x \in \F_2^n.
\end{equation}
We write
\begin{equation}
\Delta_{x^{(1)}, \ldots, x^{(k)}}f(x) = \prod_{S \subseteq [k]} f\left(x + \sum_{i\in S}x^{(i)}\right),
\end{equation}
where $x^{(i)} \in \F_2^n$, for all $i \in [k]$, and some $k \in \Z^+$.
In Equation \eqref{deriv-defn} we have defined derivatives of a Boolean function when the codomain of the function is $\{1, -1\}$. In that case, the resulting derivative turns out to be function from $\F_2^n$ to $\{1, -1\}$. The derivative of a Boolean function  $F : \F_2^n \rightarrow \F_2$,  at a point $a \in \F_2^n$ is 
\begin{equation}
\label{bf_d1}
\Delta_{a}F(x) = F(x) + F(x+a), \mbox{ for all } x \in \F_{2}^{n}.
\end{equation}
For any $a, b \in \F_2^n$
\begin{equation}
\label{bf_d2}
\Delta_{a,b}F(x) = F(x) + F(x+b) + F(x+a) + F(x+a+b).
\end{equation}
For $a, b, c \in \F_2^n$, 
\begin{equation}
\label{bf_d3}
\begin{split}
\Delta_{a,b,c}F(x) =& F(x) + F(x+c) + F(x+b) + F(x+b+c) + F(x+a) \\
& \quad + F(x+a+c) + F(x+a+b) + F(x+a+b+c).
\end{split}
\end{equation}
In general for $x^{(1)}, \ldots, x^{(k)} \in \F_2^n$, 
\begin{equation}
\label{bf_d4}
\Delta_{x^{(1)}, \ldots, x^{(k)}}F(x) = 
\sum_{S \subseteq [k]} F\left(x + \sum_{i \in S} x^{(i)}\right). 
\end{equation}

%\section{Gowers uniformity norms}

\subsection{Gowers uniformity norms}

Gowers \cite{Gowers2001} introduced (now, called Gowers) uniformity norms in his work on Szmer\'edi's theorem. For an introductory reading on the topic, we refer to the Ph.D. thesis of Chen~\cite{Chen}. 
The Gowers $U_k$ norm of $f \in \mathfrak B_n$, denoted by $\norm{f}_{U_k}$,  is defined as 
\begin{equation}
\label{gowers-defn}
\norm{f}_{U_k} = \left(2^{-(k+1)n} \sum_{x, x^{(1)}, \ldots, x^{(k)} \in \F_2^n}
\prod_{S \subseteq [k]}f \left( x +   \sum_{i \in S}x^{(i)}\right)\right)^{2^{-k}}. 
\end{equation}
%The Gowers $U_1$ norm is
%\begin{equation}
%\label{gowers-u1}
%\norm{f}_{U_1} =  \left(2^{-2n} 
%\sum_{x \in \F_2^n} \sum_{a \in \F_2^n} f(x) f(x+a) \right)^{\frac{1}{2}} 
%= 2^{-n} \sum_{x \in \F_2^n} f(x) .
%\end{equation}
The Gowers $U_2$ norm is 
\allowdisplaybreaks
\begin{eqnarray*}
%\label{gowers-u2}
\norm{f}_{U_2} &= &\left( 2^{-3n}             
\sum_{x \in \F_2^n} \sum_{a \in \F_2^n}\sum_{b \in \F_2^n} f(x) f(x+a)f(x+b)f(x+a+b)
 \right)^{2^{-2}} \nonumber\\
 & =& \left( 2^{-3n}             
\sum_{a \in \F_2^n} \sum_{x \in \F_2^n} f(x) f(x+a)\sum_{b \in \F_2^n} f(x+b)f(x+a+b)
 \right)^{2^{-2}}\nonumber\\
 & = &\left( 2^{-3n}             
\sum_{a \in \F_2^n} \sum_{x \in \F_2^n} f(x) f(x+a)\sum_{y \in \F_2^n} f(y)f(y+a)
 \right)^{2^{-2}} \\
 & = &\left( 2^{-n}             
\sum_{a \in \F_2^n} \sum_{x \in \F_2^n} \widehat f(x)^2 (-1)^{x\cdot a} \sum_{y \in \F_2^n} \widehat f(y)^2 (-1)^{y\cdot a} \right)^{2^{-2}} \nonumber\\
 & = &\left( 2^{-n}             
\sum_{x \in \F_2^n}  \sum_{y \in \F_2^n}  \widehat f(x)^2 \widehat f(y)^2
\sum_{a \in \F_2^n}  
 (-1)^{(x+y)\cdot a} \right)^{2^{-2}} = \left( \sum_{x \in \F_2^n}  \widehat f(x)^4 \right)^{2^{-2}}.\nonumber
\end{eqnarray*}
The Gowers $U_3$ norm is 
\begin{equation}
\label{gowers-u3}
\begin{split}
\norm{f}_{U_3} &= \Big( 2^{-4n}             
\sum_{x \in \F_2^n} \sum_{a \in \F_2^n}\sum_{b \in \F_2^n} \sum_{c \in \F_2^n} f(x) f(x+a)f(x+b)f(x+a+b) \\ & \qquad  f(x+c) f(x+a+c)f(x+b+c)f(x+a+b+c)
 \Big)^{2^{-3}}.
\end{split}
\end{equation}
Substituting the derivative in \eqref{gowers-u3}
\begin{equation}
\label{gowers-u3-1}
\begin{split}
\norm{f}_{U_3} &= \Big( 2^{-4n}             
\sum_{c \in \F_2^n} \sum_{x, a, b \in \F_2^n}   
\Delta_c f(x) \Delta_c f(x+a)\Delta_c f(x+b)\Delta_c f(x+a+b) \Big)^{2^{-3}}\\
&= \Big( 2^{-n}             
\sum_{c \in \F_2^n} \norm{\Delta_c f(x)}_{U_2}^{2^2} \Big)^{2^{-3}}. 
\end{split}
\end{equation}
In general, the Gowers $U_k$  norm of $f \in \mathfrak B_n$ is
\begin{equation}
\label{gowers-uk}
\begin{split}
\norm{f}_{U_k} &= \left(  2^{-(k+1)n} \sum_{x, x^{(1)}, \ldots, x^{(k)}\in \F_2^n} 
\prod_{S\subseteq [k]\setminus [2]} \prod_{T\subseteq [2]}  f\left(x + \sum_{i \in S}x^{(i)}
+ \sum_{j \in T}x^{(j)} \right)\right)^{2^{-k}}\\
&= \left(  2^{-(k+1)n}
\sum_{x, x^{(1)}, \ldots, x^{(k)}\in \F_2^n} 
 \prod_{T\subseteq [2]} \prod_{S\subseteq [k]\setminus [2]} f\left(x + \sum_{i \in S}x^{(i)}
+ \sum_{j \in T}x^{(j)} \right)\right)^{2^{-k}}\\
&= \left(  2^{-(k+1)n}
\sum_{ x^{(3)}, \ldots, x^{(k)}\in \F_2^n} 
 \sum_{ x^{(1)},  x^{(2)}\in \F_2^n} 
 \prod_{T\subseteq [2]} \Delta_{x^{(3)}, \ldots, x^{(k)} }
  f\left(x + \sum_{j \in T}x^{(j)} \right)\right)^{2^{-k}}\\
&=\left (  2^{-(k-2)n}
\sum_{ x^{(3)}, \ldots, x^{(k)}\in \F_2^n} 
 \norm{\Delta_{x^{(3)}, \ldots, x^{(k)}} f(x )}_{U_2}^{2^2}\right)^{2^{-k}}. 
\end{split}
\end{equation}
Equation~\eqref{gowers-uk} shows the relation between the Gowers $U_k$ norm and the
$U_2$ norms of the $(k-2)$th derivatives of $f$. 
The time complexity of computing the Gowers $U_2$ norm of a Boolean function $f \in \mathfrak B_n$ is $O(n2^{2n})$. Arguing  in the same way, the time complexity of computing Gowers $U_k$ norm is $O(n2^{kn})$.  

In this paper, we propose a quantum algorithm to estimate an upper bound of 
Gowers $U_2$ norm and based upon that, we find a quantum counterpart of the BLR linearity testing~\cite{BLR93} that tends to perform better than the 
classical version, assuming the availability of a quantum computer with sufficient number of qubits. The complexities of the quantum algorithms are independent
of the number of variables $n$, of course, again with the strong assumption of 
the availability of a fairly large quantum computer.  
%\comm{We need to get the complexity of the new algorithm, since we refer to the old one.}

\subsection{Gowers uniformity norms and approximation of Boolean functions
by low degree Boolean functions}
\label{approx}

In this section, we discuss the connection between the Gowers uniformity norms and the approximation of Boolean functions by low degree Boolean functions. The nonlinearity, denoted by $nl(f)$, of a Boolean function $f \in \mathfrak B_n$ is the minimum Hamming
distance from $f$ to all affine functions in $\mathfrak A_n$. That is
\begin{equation}
\label{nl}
nl(f) = \min \{ d_H(f, \varphi) : \varphi \in \mathfrak A_n \}.
\end{equation}
The $r$th-order nonlinearity of a Boolean function $f$, denoted by
$nl_r(f)$, is the minimum  Hamming distance from $f$ to the 
functions having algebraic degree less than or equal to~$r$. The first-order nonlinearity $nl_1(f) = nl(f)$. It is well known that (cf.~\cite{Pante})
\begin{equation}
\label{nl-fc}
nl(f) = 2^{n-1} \Big(1 -  \max_{x\in \F_2^n} |\widehat f (x)|\Big). 
\end{equation}
Carlet~\cite{Carlet08} obtained lower bounds of $r$th-order nonlinearity  of Boolean functions by using nonlinearities of their higher-order derivatives. This establishes a relationship between the $r$th-order nonlinearities of Boolean functions and Fourier coefficients of their derivatives. Gowers uniformity norms involve Fourier coefficients of higher-order derivatives \eqref{gowers-uk},  and serve the same purpose as evident from the following theorem. 
\begin{theorem}[\cite{Chen}, Fact 2.2.1]
\label{gw}
Let $k \in \Z^+$, $\epsilon > 0$. Let $P : \F_2^n \rightarrow \F_2$ be a polynomial of degree at most $k$, and $f: \F_2^n\rightarrow \R$. Suppose $\abs{2^{-n} \sum_{x\in \F_2^n}f(x)(-1)^{P(x)}} \geq \epsilon$. Then $\norm{f}_{U_{k+1}}\geq \epsilon$. 
\end{theorem}
For $k=1$, informally, this means that if, for some $f$, the norm $\norm{f}_{U_2}$ is small then its Fourier coefficients are small, and therefore $f$ has high nonlinearity.  On the other hand,

\begin{equation}
\label{u2-nl}
\begin{split}
 \norm{f}_{U_2}^{4} &= \sum_{x\in \F_2^n} \widehat f(x)^4 \\
&\leq \max_{x\in \F_2^n} |\widehat f (x)|^2 \sum_{x\in \F_2^n} \widehat f(x)^2 \\
&= \max_{x\in \F_2^n} |\widehat f (x)|^2 \mbox{ (applying Parseval identity)}\\
&= (1 - 2^{1-n}nl(f))^2 \mbox{                  (using \eqref{nl-fc})}. 
\end{split}
\end{equation}
Equation~\eqref{u2-nl} tells us that if a Boolean function has high nonlinearity then 
its $U_2$ norm is small, and if $U_2$ norm is large, then the nonlinearity is small.  

The second-order nonlinearity of a Boolean function is the minimum of the distances of that function from the quadratic Boolean function (i.e., the Boolean functions with algebraic degree at most $2$). By Theorem~\ref{gw}, for all polynomials  $P: \F_2^n \rightarrow \F_2$ of degree  at most $2$, if $\norm{f}_{U_3} < \epsilon$,  then $\abs{2^{-n} \sum_{x \in \F_2^n}f(x)(-1)^{P(x)}} < \epsilon$. Therefore, the second-order nonlinearity of such functions ought to be high. 
Green and Tao~\cite{green2005inverse} proved that just as for $U_2$, if a Boolean function has high second-order nonlinearity, then its $U_3$ norm is low. They also proved that such an implication is not valid for $U_k$ norms for $k \geq 4$. 

The discussion in this section points to the fact that Gowers $U_2$ and $U_3$ norms have the promise of being good indicators for the first and second-order nonlinearities of a Boolean function. 
Determination of these nonlinearities have complexities that scale exponentially with the number of input variables of  Boolean functions.  In the following section, we propose a quantum algorithm to estimate an upper bound of the Gowers $U_2$  norm that is probabilistic in nature, the probability converges as $e^{-2m^2t^2}$ where $m$ is the number of trials and $t$ is a positive error margin.

\subsection{Quantum information: definitions and notation}

In this section, we will introduce some notation that we use throughout the paper. For an introduction to quantum computing, we refer to Rieffel and Polak \cite{Rieffel}, or Nielsen and Chuang \cite{NC10}.

 A {\em qubit} or {\em qu-bit}  can be described by a vector $\Ket{\psi}=(a,b)^T\in{\mathbb{C}}^2$, where `$T$' indicates the transpose, $|a|^2$
is the probability of observing the value 0 when we measure
the qubit, and $|b|^2$ is the probability of observing 1. If both $a$ and
$b$ are nonzero, the qubit has both the value 0 and 1 at the same
time, and we call this a {\em superposition}. Once we have measured the
qubit, however, the superposition collapses, and we are left with a
classical state that is either 0 or 1 with certainty. A state of $n$ qubits
is represented by a normalized complex vector with $2^n$ elements. We define $\Bra{\psi}$ as the conjugate transpose of $\Ket{\psi}$. This notation is known as the bra-ket notation. We denote the standard basis (column) vectors as  $\ket0$ and $\ket1$, and then $\Ket{\psi}=(a,b)^T=a\ket0+b\ket1$.

In the following, we will use the conventional notation $\ket a \ket b:=\ket a\otimes \ket b$,  or $\Ket{ a  b}:=\ket a\otimes \ket b$. A state on $n$ qubits can be represented as a $\C$-linear combination of the vectors of the standard basis $\ket\psi=\sum_{x\in\F_2^n}a_x\ket x$, where $a_x\in\C,\ \forall x\in \F_2^n$, and $\sum_{x\in\F_2^n}|a_x|^2=1$.

 Let $\ket{0_n}$ be the quantum state associated with the zero vector in $\F_2^n$.  Let $\ket + = \frac{\ket 0 + \ket 1}{\sqrt 2}$ and $\ket - = \frac{\ket 0 - \ket 1}{\sqrt 2}$. For any $x \in \F_2^n$ and $ \epsilon \in \F_2$, the {\it bit oracle implementation} $U_F$ of $F$ is 
\begin{equation}
\label{bit-oracle}
\ket \varepsilon \ket x \xrightarrow{U_F} \ket{\varepsilon + F(x)} \ket x.  
\end{equation}

Here, $n$-qubits in $\ket{x}$ specify the input state that changes the target qubit $\ket{\epsilon}$ according to the value of the Boolean function $F(x)$.

If the first qubit is $\ket -$, then $\ket - \ket x \xrightarrow{U_F} (-1)^{F(x)}\ket - \ket x$.  We write $\ket x \xrightarrow{U_F} (-1)^{F(x)} \ket x$ with the understanding that there is an additional target qubit in the $\ket{-}$ state that remains unchanged and refer to this as the  {\it phase oracle implementation} of the function~$F$. Suppose that a computational basis state is of the form $\ket{x^{(1)}\| x^{(2)} \| \ldots \|x^{(m)}}$ where for any two vectors $x \in \F_2^r$ and $y \in \F_2^s$, the concatenation 
$x \| y = (x_1, \ldots, x_r, y_1, \ldots, y_s)$ is a vector in $\F_2^{r+s}$. 
It is reasonable to write $\ket{x^{(1)}\| x^{(2)} \| \ldots \|x^{(m)}} = \ket{x^{(1)}} \ket{x^{(2)}}  \ldots \ket{x^{(m)}}$. The vector $x^{(i)} \in \F_2^{r_i}$, for some $r_i \in \Z^+$ is said to be the content of the 
$i$th register. If $x^{(i)}, x^{(j)} \in \F_2^r$, for some $r \in \Z^+$, we define 
$MCNOT_i^j$ as 
%a multi-qubit CNOT operation that operates on the corresponding qubits of the vectors in the registers $i, j$ respectively.
%\begin{equation}
%MCNOT_i^j \,\equiv\, \bigotimes_{k = 1}^{r}CNOT_{x^{(j)}_k}^{x^{(i)}_k}
%\end{equation} 
%The action of this gate on the $m$ register state is given by,
\begin{equation*}
\label{cnot-def}
\ket{x^{(1)}} \ldots \ket{x^{(i)}}  \ldots \ket{x^{(j)}}  \ldots \ket{x^{(m)}}
\xrightarrow{MCNOT_i^j} \ket{x^{(1)}} \ldots \ket{x^{(i)}+x^{(j)}}  \ldots \ket{x^{(j)}}  \ldots \ket{x^{(m)}}. 
\end{equation*}
We can realize the transformation induced by $MCNOT_i^j$ by using an appropriate number of conventional $CNOT$ gates. 

Let $I = \begin{pmatrix}1&0\\0&1 \end{pmatrix} $ be the $2\times 2$ identity matrix, and $H = \frac{1}{\sqrt{2}} \begin{pmatrix*}[r] 1 & 1 \\ 1 & -1 \end{pmatrix*}$ be the $2\times 2$ Hadamard matrix. The tensor product of matrices is denoted by $\otimes$. The matrix $H_n$ is recursively defined as: 
\begin{equation}
\label{def-hadamard}
\begin{split}
H_2 &= H \otimes H, \\
H_n &= H_{n-1} \otimes H_{n-1}, \mbox{ for all } n \geq 3.
\end{split}
\end{equation}
Note that, for $x\in\F_2^n$, $H_n\ket x=2^\frac{-n}{2}\sum_{x'\in\F_2^n}(-1)^{x\cdot x'}\ket{x'}$.

In the next section we propose an algorithm to compute Gowers $U_2$ norm of Boolean functions. Our approach   resembles that employed by Bera, Maitra, and Tharrmashashtha~\cite{BeraMT19} to estimate the autocorrelation spectra of Boolean functions. 

\section{A quantum algorithm to estimate Gowers uniformity norms}

%\subsection{Computing the Gowers $U_2$ norm}
We prepare the quantum state 
$2^{-\frac{3n}{2}} \sum_{b \in \F_2^n} \sum_{a \in \F_2^n}\sum_{x \in \F_2^n} \ket{x}\ket{a}\ket{b}$, and apply the following transformations: 
\allowdisplaybreaks
\begin{eqnarray}
\label{main-algo-1}
&&2^{-\frac{3n}{2}} \sum_{b \in \F_2^n} \sum_{a \in \F_2^n}\sum_{x \in \F_2^n} \ket{x}\ket{a}\ket{b}\nonumber\\
&&\xrightarrow{U_F \otimes I \otimes I}  
2^{-\frac{3n}{2}} \sum_{b \in \F_2^n}  \sum_{a \in \F_2^n}\sum_{x \in \F_2^n} (-1)^{F(x)}\ket{x}\ket{a}\ket{b}\nonumber\\
&&\xrightarrow{MCNOT_1^2}  
2^{-\frac{3n}{2}} \sum_{b \in \F_2^n}  \sum_{a \in \F_2^n}\sum_{x \in \F_2^n} (-1)^{F(x)}\ket{x+a}\ket{a}\ket{b}\nonumber\\
&&\xrightarrow{U_F \otimes I \otimes I}  
2^{-\frac{3n}{2}} \sum_{b \in \F_2^n}  \sum_{a \in \F_2^n}\sum_{x \in \F_2^n} (-1)^{F(x) + F(x+a)}\ket{x+a}\ket{a}\ket{b}\nonumber\\
&&\xrightarrow{MCNOT_1^2}  
2^{-\frac{3n}{2}} \sum_{b \in \F_2^n}  \sum_{a \in \F_2^n}\sum_{x \in \F_2^n} (-1)^{F(x)+ F(x+a)}\ket{x}\ket{a}\ket{b}\nonumber\\
&&\xrightarrow{MCNOT_1^3}  
2^{-\frac{3n}{2}} \sum_{b \in \F_2^n}  \sum_{a \in \F_2^n}\sum_{x \in \F_2^n} (-1)^{F(x)+ F(x+a)}\ket{x+b}\ket{a}\ket{b}\\
&&\xrightarrow{U_F \otimes I \otimes I}  
2^{-\frac{3n}{2}} \sum_{b \in \F_2^n}  \sum_{a \in \F_2^n}\sum_{x \in \F_2^n} (-1)^{F(x)+F(x+a)+F(x+b)}\ket{x+b}\ket{a}\ket{b}\nonumber\\
&&\xrightarrow{MCNOT_1^2}  
2^{-\frac{3n}{2}} \sum_{b \in \F_2^n}  \sum_{a \in \F_2^n}\sum_{x \in \F_2^n} (-1)^{F(x)+F(x+a)+F(x+b)}\ket{x+a+b}\ket{a}\ket{b}\nonumber\\
&&\xrightarrow{U_F \otimes I \otimes I}  
2^{-\frac{3n}{2}} \sum_{b \in \F_2^n}  \sum_{a \in \F_2^n}\sum_{x \in \F_2^n} (-1)^{\Delta_{a,b}F(x)}\ket{x+a+b}\ket{a}\ket{b}\nonumber\\
&&\xrightarrow{MCNOT_1^2}  
2^{-\frac{3n}{2}} \sum_{b \in \F_2^n}  \sum_{a \in \F_2^n}\sum_{x \in \F_2^n} (-1)^{\Delta_{a,b}F(x)}\ket{x+b}\ket{a}\ket{b}\nonumber\\
&&\xrightarrow{MCNOT_1^3}  
2^{-\frac{3n}{2}} \sum_{b \in \F_2^n}  \sum_{a \in \F_2^n}\sum_{x \in \F_2^n} (-1)^{\Delta_{a,b}F(x)}\ket{x}\ket{a}\ket{b}\nonumber\\
&&\xrightarrow{H_n^{\otimes 3}}   
 \sum_{a', b', x' \in \F_2^n}
(2^{-3n} \sum_{x, a, b \in \F_2^n} (-1)^{\Delta_{a,b}F(x) + a\cdot a' + b \cdot b' + x \cdot x'})\ket{x'}\ket{a'}\ket{b'}.\nonumber
\end{eqnarray}

 It should be remembered that in addition to the three $n$-qubit registers used, there is an additional target qubit that is in the $\ket{-}$ state and remains unchanged throughout, owing to this fact the qubit has been dropped from the sequence of operations, for brevity.

The above sequence of operations excluding the last $H_n^{\otimes 3}$ is summarized as 
\begin{equation}
\label{main-algo-2}
2^{-\frac{3n}{2}} \sum_{a, b, x \in \F_2^n } \ket{x, a, b} \xrightarrow{\mathfrak D_F}2^{-\frac{3n}{2}} \sum_{a, b, x \in \F_2^n} (-1)^{\Delta_{a,b}F(x)}\ket{x, a, b}.
\end{equation}
The probability that a measurement of the resultant state yields the result $\ket{0_n}\ket{0_n}\ket{0_n}$ is given by
\begin{equation*}
\Pr[ x' = a' = b' = 0_n] = \Big(2^{-3n}\sum_{x,a,b \in \F_2^{n}}{(-1)}^{\Delta_{a,b}F(x)}\Big)^2
\end{equation*}
and since $f(x) = {(-1)}^{F(x)}$, using \eqref{bf_d2}
\begin{equation}
\label{qg_U2result}
\Pr[ x' = a' = b' = 0_n] = ({\norm{f}}_{U_2}^{4})^2 = \norm{f}_{U_2}^{8}. 
\end{equation}

\subsection{Estimation of the upper bound of Gowers $U_2$ norm}

Let the final output state at the end of the transformation described in \eqref{main-algo-1} be 
\begin{equation}
\ket{\Psi} =   \sum_{a', b', x' \in \F_2^n} C(x', a', b') \ket{x' a' b'}.
\end{equation}
The  probability amplitude of the state $\ket{x' a' b'}$ is 
\begin{equation}
C(x', a', b') = 2^{-3n} \sum_{x, a, b \in \F_2^n} (-1)^{\Delta_{a,b}F(x) + a\cdot a' + b \cdot b' + x \cdot x'}.
\end{equation}
The outcome of a measurement, with respect to the computational basis, performed on the output state is a $3n$ bit string $(x'\|a'\|b')$, where $x', a', b' \in \F_2^n$, and the probability of measuring said string is $\abs{C(x', a', b')}^2$. Therefore, the eighth power of the Gowers $U_2$ norm is given by $\abs{C(0_n, 0_n, 0_n)}^2$.
The next theorem outlines a strategy to determine a probabilistic upper bound of the Gowers $U_2$ norm. 
\begin{theorem}
\label{th-g-u2}
We assume that the measurements are done with respect to the computational basis. 
Suppose that $Y$ is a random variable defined on the set of all possible measurement outcomes on the quantum state $H_n^{\otimes 3} \circ \mathfrak D_F \Big(2^{-\frac{3n}{2}} \sum_{a, b, x \in \F_2^n } \ket{x}\ket{a}\ket{b}\Big)$  as 
\begin{equation*}
Y(x', a', b') = 2^{-3n}(x'\|a'\|b')_{10},
\end{equation*} 
where $(x'\|a'\|b')_{10}$ is the decimal value of the concatenated $3n$ bit string. Following the usual convention, we write $Y$ instead of $Y(x', a', b')$.
Suppose that $(Y_1, \ldots, Y_m)$ be a random sample such that each $Y_i$ is independent and identically distributed as $Y$. Let $\overline{Y} = \frac{1}{m}\sum_{i\in [m]} Y_i$. 
Then 
$$\Pr\left[ \norm{f}_{U_2} \leq (1 + t -\overline{Y})^{1/2^3}  \right] \geq 1 - \exp(-2m^2t^2),$$
%\comm{from \eqref{qg_U2result}, the eighth root should be taken as opposed to the reciprocal of the eighth power.}
for any positive real number $t$. 
\end{theorem}
\begin{proof}
 Let the expectation of $Y$, 
$E[Y] = \mu$. %Consider two events, 
%$E_0 = \{ Y : Y= 0 \}$ and  $E_{1} = \{ Y : 0 < Y \leq 1\}$. 
Let $\Pr[Y=0] =  \norm{f}_{U_2}^{8} = p$, so $\Pr[Y \neq 0] = 1-p$.
The range of the random variable $Y$ has $2^{3n}$ distinct values in the interval $[0, 1]$
including $0$. Let us denote them by $y_0, y_1, \ldots, y_{2^{3n}-1}$, where $y_j=2^{-3n}j$. The expectation of $Y$ is
\begin{equation}
\label{exp-inq}
\begin{split}
\mu = E[Y] &=  y_0 \Pr[Y=0] + y_1 \Pr[Y=y_1] + \cdots + y_{2^{3n}-1}\Pr[Y=y_{2^{3n}-1}]\\
&=  y_1 \Pr[Y=y_1] + y_2 \Pr[Y=y_2] + \cdots + y_{2^{3n}-1} \Pr[Y=y_{2^{3n}-1}]\\
&<   \Pr[Y=y_1] + \Pr[Y=y_2] + \cdots + \Pr[Y=y_{2^{3n}-1}]\\
&= \Pr[Y \neq 0] = 1-p.  
\end{split}
\end{equation} 
%\comm{The upper bound is arrived at using the fact that each $y_j < 1$. This bound is also independent of $n$}
Suppose that $(Y_1, \ldots, Y_m)$ be a random sample of size $m$. The sample mean 
is $\overline{Y} = \frac{1}{m} \sum_{i \in [m]} Y_i$. By the Hoeffding's inequality
\cite{Hoeffding1963} 
\begin{equation}
\label{Hoeffding}
\Pr\left[\overline{Y} \geq \mu + t\right] \leq \exp(-2m^2t^2). 
\end{equation}
where $t$ is any positive real number. Using equations \eqref{exp-inq} and \eqref{Hoeffding},
\begin{equation}
\label{estimate}
\begin{split}
& \Pr\left[ 1-p > \mu \geq \overline{Y} - t \right] \geq 1 - \exp(-2m^2t^2), \\
\mbox{which implies }& \Pr\left[ p < 1 + t -\overline{Y}  \right] \geq 1 - \exp(-2m^2t^2),\\
\mbox{that is}, & \Pr\left[ \norm{f}_{U_2} < (1 + t -\overline{Y})^{1/2^3}  \right] \geq 1 - \exp(-2m^2t^2).  
\end{split}
\end{equation}
The theorem is shown.\qed
\end{proof}
The last line of \eqref{estimate} tells us that if we measure $m$ times and compute $\overline{Y}$, then the probability that $\norm{f}_{U_2}$ is bounded above by $(1 + t -\overline{Y})^{1/2^3}$ is $1 - \exp(-2m^2t^2)$. Therefore with an appropriate choice of $m$ and $t$ we can estimate an upper bound of the Gowers $U_2$ norm of $f$ with a very high probability.

\subsection{Linear approximation employing the Gowers $U_2$ norm}
\label{lin-approx}

We start by defining distance between Boolean functions in terms of 
probabilities. 
\begin{definition}
For any two functions $f, g \in \mathfrak B_n$, 
$$\dist(f, g) = \Pr_{\mathbf x \sim \F_2^n}[f(\mathbf x) \neq g(\mathbf x)]= \frac{d_H(f,g)}{2^n}$$
where $\mathbf x$ is a random variable uniformly distributed over $\F_2^n$.
\end{definition}
The function $f$ is said to be $\epsilon$-close to $g$ if 
$\dist(f, g) \leq \epsilon$, and $\epsilon$-far from $g$ if 
$\dist(f, g) > \epsilon$. 
We will now design an algorithm to determine whether a function is linear or $\epsilon$-far from linear; we refer to Hillery and Anderson \cite[Section III]{HilleryA2011} for a discussion on such tests. 

{
\center
\begin{algorithm}[!ht]
\centering
\begin{algorithmic}[1]
\STATEx{Input: Quantum implementation of $f \in \mathfrak B_n$. }
\STATE{Initial state: $2^{-\frac{3n}{2}} \sum_{a, b, x \in \F_2^n } \ket{x, a, b}$. }
\STATE{Perform the following sequence of transformations:   \begin{equation*}
\begin{split}
&2^{-\frac{3n}{2}} \sum_{a, b, x \in \F_2^n } \ket{x, a, b}\\ &\xrightarrow{\mathfrak D_F}2^{-\frac{3n}{2}} \sum_{a, b, x \in \F_2^n} (-1)^{\Delta_{a,b}F(x)}\ket{x, a, b} \\
&\xrightarrow{H_n^{\otimes 3}}   
 \sum_{a', b', x' \in \F_2^n}
(2^{-3n} \sum_{x, a, b \in \F_2^n} (-1)^{\Delta_{a,b}F(x) + a\cdot a' + b \cdot b' + x \cdot x'})
\ket{x',a',b'}.
\end{split}
\end{equation*}
 }
\STATE{Measure the output state with respect to the computational basis.} 
\STATE{If the measurement result is $\ket{0_n, 0_n,0_n}$ then ``ACCEPT'' (the function is linear). } 
\STATE{Else ``REJECT''.} 
\end{algorithmic}
\caption{Linearity checking with the Gowers $U_2$ norm.}
\label{algo-lin}
\end{algorithm}
}
\begin{theorem}
\label{testing-linearity}
If $f$ is a linear function then the output is ``ACCEPT'' with probability~$1$. 
If $f$ is $\epsilon$-far from linear functions, then probability of ``REJECT'' is greater than 
$1 - \exp(-8\epsilon)$.
\end{theorem}
\begin{proof}
If $f$ is a linear functions, then the output is ``ACCEPT'' with certainty. This directly 
follows from the definition of Gowers $U_2$ norm. 
If $f$ is $\epsilon$-far from linear functions, then
\begin{equation*}
\norm{f}_{U_2}^{2^3} \leq \left(1 - 2 \frac{nl(f)}{2^n}\right)^4 \leq (1 - 2 \epsilon)^4. 
\end{equation*}
This means that the probability that the output is ``ACCEPT'' is less than or equal 
to $(1 - 2 \epsilon)^4$; therefore the probability of ``REJECT'' is greater than 
$1- (1 - 2 \epsilon)^4 \approx 1 - \exp(-8\epsilon)$.  \qed
\end{proof}
The result concerning the BLR test is: 
\begin{theorem}{\textup{\cite[Theorem 1.30]{O'Donnell}}}
\label{blr}
Suppose the BLR Test accepts $F: \F_2^n \rightarrow \F_2$ with probability $1 - \epsilon$. Then 
$f$ is $\epsilon$-close to being linear. 
\end{theorem}
By the BLR test, if a function is $\epsilon$-far from the linear functions, and
it is promised that we have such functions and linear functions only,
then given a function from the latter class, the probability that the algorithm will  REJECT is greater than $\epsilon$. 

\begin{remark}
The algorithm presented here has been implemented in the IBM quantum machine (https://www.ibm.com/quantum-computing/) for some small examples and has given the expected output of probabilities.
\end{remark}

\section{Appendix: generalization to higher Gowers norms}

The same technique can be used for other Gower's norms. For instance, we can apply the unitary transformation $H_n^{\otimes 4}\circ \mathfrak D_F^3$ to the state $2^{-2n}\sum_{x\in\F_2^n}\sum_{a\in\F_2^n}\sum_{b\in\F_2^n}\sum_{c\in\F_2^n}\ket x\ket a\ket b\ket c$, where, with notation $M_i^j=MCNOT_i^j$ and $U_F^3=U_F\otimes I\otimes I\otimes I$,
%$$\mathfrak D_F^3=MCNOT_1^3\circ (U_F\otimes I\otimes I\otimes I)\circ MCNOT_1^3\circ MCNOT_1^4(U_F\otimes I\otimes I\otimes I)\circ MCNOT_1^3\circ (U_F\otimes I\otimes I\otimes I)\circ MCNOT_1^2\circ (U_F\otimes I\otimes I\otimes I)\circ MCNOT_1^4\circ(U_F\otimes I\otimes I\otimes I)\circ MCNOT_1^3\circ(U_F\otimes I\otimes I\otimes I)\circ MCNOT_1^2\circ(U_F\otimes I\otimes I\otimes I)$$
$$\mathfrak D_F^3=M_1^3\circ U_F^3\circ M_1^3\circ M_1^4\circ U_F^3\circ M_1^3\circ U_F^3\circ M_1^2\circ U_F^3\circ M_1^4\circ U_F^3\circ M_1^3\circ U_F^3\circ M_1^2\circ U_F^3.$$
We obtain thus the state $\sum_{a',b',b',x'\in\F_2^n}2^{-4n}\sum_{a,b,c,x\in\F_2^n}(-1)^{\Delta_{a,b,c}F(x)}\ket{x,a,b,c}$. Then, $Pr[x'=a'=b'=c'=0_n]=\left(2^{-4n}\sum_{a,b,c,x\in\F_2^n}(-1)^{\Delta_{a,b,c}F(x)}\right)^2$, and, using \eqref{bf_d3}, $Pr[x'=a'=b'=c'=0_n]=\left(\norm{f}_{U_3}^8\right)^2=\norm{f}_{U_3}^{16}$.

\vskip.6cm
\noindent
{\bf Acknowledgment:} Research of C.~A.~Jothishwaran and Sugata Gangopadhyay is a part of the project ``Design and Development 
of Quantum Computing Toolkit and Capacity Building'' sponsored by 
the Ministry of Electronics and Information Technology (MeitY) of the 
Government of India.

\bibliographystyle{splncs03}
%\bibliography{sugo-cryptoref}

\end{document}